\documentclass[submission,copyright,creativecommons]{eptcs}
\providecommand{\event}{Refine 2013} 

\usepackage{color}
\usepackage{url}
\usepackage[latin1]{inputenc}
\usepackage{amsmath}
\usepackage{amscd}
\usepackage{latexsym}
\usepackage{stmaryrd}
\usepackage{amssymb}
\usepackage{amscd}
\usepackage{fancyvrb}
\usepackage{boxedminipage}

\ifpdf
\usepackage[pdftex]{graphicx}
\else
\usepackage{graphicx}
\fi
 
 \usepackage[all]{xy}

 \def\eg{e.g.\,}

\usepackage{pdfsync}
\usepackage{latexsym,alltt,fleqn}
\usepackage{graphics}
\usepackage{epsfig}
\usepackage{times}
\usepackage{amsmath}
\usepackage[ps,curve,2cell]{xypic}
\UseAllTwocells
\usepackage[active]{srcltx}
\usepackage{xy}
\xyoption{v2}
\usepackage{dsfont}
\usepackage{color}

\newcommand{\arity}{\underline{\mathrm{ar}}}

\newcommand{\st}{\mathrm{st}}
\newcommand{\op}{\mathrm{op}}

\def\HI{\mathcal{HI}}

\newcommand{\Nom}{\mathrm{Nom}}
\newcommand{\MOD}{\mathrm{Mod}^{\HI}}
\newcommand{\Mod}{\mathrm{Mod}}
\newcommand{\SEN}{\mathrm{Sen}^{\HI}}
\newcommand{\Sen}{\mathrm{Sen}}

\newcommand{\Sign}{\mathrm{Sign}}
\newcommand{\SIGN}{\mathrm{Sign}^{\HI}}

\newcommand{\MS}{\mathrm{MS}}
\newcommand{\REL}{\REL}
\newcommand{\Sig}{\mathrm{Sig}}

\def\ie{i.e.}

\def\Set{\mathbb{S}et}
\def\CAT{\mathbb{C}AT}
\def\PL{\mathit{PL}}
\def\REL{\mathit{REL}}
\def\MVL{\mathit{MVL}}

\def\EQ{\mathit{EQ}}

\def\HPL{\mathcal{H}\PL}

\def\H2PL{\mathcal{H}^2\PL}

\def\HMVL{\mathcal{H}\MVL}

\def\HEQ{\mathcal{H}\mathit{EQ}}

\newcommand{\w}{w}

\newcommand{\SenI}{\mathrm{Sen}^{\mathcal{I}}}
\newcommand{\SignI}{\mathrm{Sign}^{\mathcal{I}}}
\newcommand{\ModI}{\mathrm{Mod}^{\mathcal{I}}}

\newcommand{\bisim}{\rightleftharpoons_\varphi}
\newcommand{\bis}{\mathrm{B}_\varphi}
\newcommand{\simf}{\mathrm{R}_\varphi}

\newcommand{\feq}{\gg_\varphi}

\def\just#1#2{\\
         &#1& \rule{2em}{0pt} \footnotesize{ \{ \mbox{\rule[-.7em]{0pt}{1.8em} #2} \}} \\ && }

\def\implies{\mathbin{\Rightarrow}}

\def\rcb#1#2#3#4{\def\nothing{}\def\range{#3}\mathopen{\langle}#1 \ #2 \ \ifx\range\nothing::\else: \ #3 :\fi \ #4\mathclose{\rangle}}

\def\just#1#2{\\
         &#1& \rule{2em}{0pt} \footnotesize{ \{ \mbox{\rule[-.7em]{0pt}{1.8em} #2} \}} \\ && }

\newcommand{\I}{\mathcal{I}}
\newcommand{\co}{\,\colon\;}
\newcommand{\ra}{\rightarrow}

\newtheorem{definition}{Definition}[section]
\newtheorem{example}{Example}[section]
\newtheorem{theorem}{Theorem}[section]

\newtheorem{corollary}{Corollary}[section]

\newenvironment{proof}{\noindent {\bf Proof.}~}{\hfill$\Box$
\medskip
}

\DeclareMathAlphabet{\mathbb}{U}{msb}{m}{n}
\DeclareSymbolFont{ams}{U}{msa}{m}{n}
\DeclareSymbolFontAlphabet{\mathams}{ams}
\DeclareMathSymbol{\filter}{\mathams}{ams}{22}

\allowdisplaybreaks

\title{Bisimilarity and refinement for hybrid(ised) logics}
\author{
Alexandre Madeira 
\institute{HASLab - INESC TEC \& Univ. Minho\\ Dep. Mathematics, Univ. Aveiro,  Portugal\\
Critical Software S.A., Portugal} 
\email{madeira@ua.pt} 
\and
Manuel A. Martins 
\institute{Dep. Mathematics, \\ Univ. Aveiro,  Portugal} 
\email{martins@ua.pt} 
\and
Luís S.~Barbosa
\institute{HASLab - INESC TEC \& Univ. Minho \\ Braga, Portugal}
\email{lsb@di.uminho.pt}
}
\def\titlerunning{Bisimilarity and refinement for hybrid(ised) logics}
\def\authorrunning{A. Madeira, M.A. Martins \& L.S. Barbosa}
\begin{document}
\maketitle

\begin{abstract}

The complexity of modern software systems entails the need for
reconfiguration mechanisms governing the dynamic evolution of their
execution configurations in response to both external stimulus or
internal performance measures. Formally, such systems may be
represented by transition systems whose nodes correspond to the
different configurations they may assume. Therefore, each node is
endowed with, for example, an algebra, or a first-order structure, to
precisely characterise the semantics of the services provided in the
corresponding configuration.  

Hybrid logics, which add to the modal description of transition structures the ability to refer to specific states, offer a generic framework to approach the specification and design of this sort of systems. Therefore, the quest for suitable notions of equivalence and refinement between models of hybrid logic specifications becomes fundamental to any design discipline adopting this perspective. This paper contributes to this effort from a distinctive point of view: instead of focussing on a specific hybrid logic, the paper introduces notions of bisimilarity and refinement for hybridised logics, i.e. standard specification logics (e.g. propositional, equational, fuzzy, etc) to which modal and hybrid features were added in a systematic way.

\end{abstract}

\section{Introduction}

The qualifier \emph{reconfigurable} is used for software systems which
behave differently in different modes of operation (often called
\emph{configurations}) and commute between them along their
lifetime. Formally, such different behaviours can be modelled by
imposing additional structure upon states in a transition system
expressing the overall system's dynamics. This path has been explored
in the authors' recent work \cite{sefm11} on a specification
methodology for reconfigurable systems. The basic insight is that,
starting from a classical state-machine specification, each state,
regarded as a possible system's configurations,  is equipped with a rich
mathematical structure to describing its
 functionality. Technically, specifications become
structured state-machines, states denoting algebras or first order
structures, rather than sets. 

A specification for this sort of system, as discussed in \cite{sefm11}, should be able to make assertions both about the transition dynamics and, locally, about each particular configuration. This leads to the adoption of hybrid logic \cite{livro_brauner} as the specification \emph{lingua franca} for the envisaged methodology.

 However, because specific problems may require specific logics to describe their configurations (\eg, equational, first-order, fuzzy, etc.), our approach is rooted on very general grounds. Instead of choosing a particular version of hybrid logic, we start by choosing a specific logic for expressing requirements at the configuration (static) level. This is later taken as the \emph{base} logic on top of which the characteristic features of hybrid logic, both at the level of syntax (i.e. modalities, nominals, etc.) and of the semantics (i.e. possible worlds), are developed. This process is called \emph{hybridisation} and was characterised in \cite{calco,paperdiaco} as well as in the first author's forthcoming PhD thesis \cite{madeirathesis}. To be completely general, the approach to hybridisation is framed in the context of the institution theory of Goguen and Burstall \cite{ins,livrodiaconescu}, each logic (base and hybridised) treated abstractly as an \emph{institution}.

In this context,  the quest for suitable notions of \emph{equivalence} and \emph{refinement} between models of hybridised logic specifications becomes fundamental to the envisaged design methodology. Such is the purpose of the present paper. Its contribution is a characterisation of bisimilarity and refinement for hybridised logics which requires a form of \emph{elementary equivalence} \cite{Hodges:1997:SMT:262326} between bisimilar states,
as a generic formulation of the usual informal requirement that \emph{truth remains invariant}. Clearly what \emph{elementary equivalent} means in each case boils down to the way the satisfaction relation is defined for the base logic used to specify the semantics of local configurations.  

The choice of similarity and bisimilarity to base refinement and equivalence of (models of) reconfigurable systems seems quite standard as a fine grained approach to observational methods for systems comparison. The notion of bisimulation and the associated conductive proof method, which is now pervasive in Computer Science, originated in concurrency theory due to the seminal work of David Park \cite{Par81} and R. Milner  in the quest for an appropriate definition of observational equivalence for communicating processes.
But the concept also arose independently in modal logic as a refinement of notions of homomorphism between algebraic models.
In the sequel the concept is revisited for models of hybridised logics adding up to the design methodology mentioned above.

The paper  is organized as
follows: Section \ref{background} recalls institutions as abstract characterisations of logics and provides a brief, and simplified, overview of the 
hybridization method proposed in \cite{calco,paperdiaco}. This forms the context for the paper's contribution. Then, Section~\ref{sec-bis} introduces 
 a general 
notion of bisimulation for hybridised logics and  characterizes the
preservation of logic satisfaction under it.   Section
\ref{sec-sim} follows a similar path but focussing  on refinement as witnessed by a simulation relation.

\section{Background}\label{background}
\subsection{Institutions}
An \emph{institution} is a category theoretic
formalisation\footnote{The language of category theory  \cite{MacLaneS:catwm}  is used to set the scene for institutions; categories, however, play no role in the paper's contribution.} of a logical system, encompassing
syntax, semantics and satisfaction. The concept was put forward by Goguen and
Burstall, in the end of the 
seventies, in order to \emph{``formalise the formal notion of logical
systems''}, in response  to the \emph{``population explosion among the
logical 
systems used in Computing Science''} \cite{ins}. 

The universal character of  institutions proved effective and resilient as witnessed by the 
 wide number of logics formalised in this framework.
Examples range from the usual logics in
classical mathematical logic (propositional, equational, first order, etc.), to the ones underlying 
specification and programming languages or used  for describing particular systems from different domains.
Well-known examples include
 \emph{probabilistic logics}
 \cite{DBLP:conf/wcii/BeierleK02},  \emph{quantum logics}
 \cite{DBLP:conf/birthday/CaleiroMSS06}, \emph{hidden and observational
   logics} \cite{hiding,DBLP:journals/jlp/BidoitH06}, \emph{coalgebraic
 logics} \cite{DBLP:journals/jlp/Cirstea06},  as well as logics for reasoning about \emph{process
 algebras} \cite{Mossakowski:2006:SCP:1763794.1763801}, 
 \emph{functional}
 \cite{livro_sannella,DBLP:journals/tcs/SchroderM09}
 and \emph{imperative programing 
 languages} \cite{livro_sannella}.
 
The theory of institutions (see \cite{livrodiaconescu} for a extensive account) was 
motivated by the need to abstract from the particular
details of each individual logic and characterise generic issues, such as satisfaction and combination of logics, in very general terms.
In Computer Science, this lead to  the development of a solid
\emph{institution-independent specification theory}, on which,  structuring and parameterisation mechanisms, required to scale up software specification methods,
are defined `once and for all', irrespective of the concrete logic used in each application domain.
 The definition is recalled below (e.g., \cite{ins,livrodiaconescu})
 and illustrated with a few examples 
 to which we  return later in the paper.

\begin{definition}[Institution]\label{ins-dfn}
An  \emph{institution}
\[\I=\big(\SignI, \SenI, \ModI, (\models^\mathcal{I}_\Sigma)_{\Sigma \in
  |\SignI|}\big)\]
consists of 
\begin{itemize}\itemsep -1pt
\item a category $\SignI$ whose objects are called \emph{signatures} and arrows  \emph{signature} morphisms;
\item a functor $\SenI\co\SignI \rightarrow \Set$  giving for each 
signature a set whose elements are called \emph{sentences} over that signature;
\item\label{itemmodels}  a functor $\ModI\co(\SignI)^{op}\rightarrow \CAT$, 
giving for each signature $\Sigma$ a category whose objects are called 
\emph{$\Sigma$-models}, and whose arrows are called 
\emph{$\Sigma$-(model) homomorphisms}; each arrow
$\varphi:\Sigma\rightarrow \Sigma' \in \SignI$, (i.e.,
$\varphi:\Sigma'\rightarrow \Sigma\in (\SignI)^{op}$) is mapped into a
functor $\ModI(\varphi):\ModI(\Sigma')\rightarrow \ModI(\Sigma)$
called a \emph{reduct functor}, whose effect is to cast a model of $\Sigma'$ as a model of $\Sigma$;
\item a relation 
$\models_\Sigma^\mathcal{I}\subseteq |\ModI(\Sigma)|\times \SenI(\Sigma)$ 
for each $\Sigma \in |\SignI|$, called the \emph{satisfaction relation}, 
\end{itemize}
such that for each morphism $\varphi\co\Sigma \rightarrow \Sigma' \in \SignI$, 
the satisfaction condition
\begin{equation}
M'\models^{\mathcal{I}}_{\Sigma'} \SenI(\varphi)(\rho)\; \text{ iff }\;
\ModI(\varphi)(M') \models^\mathcal{I}_\Sigma \rho
\end{equation}
holds for each $M'\in |\ModI (\Sigma')|$ and $\rho \in \SenI
(\Sigma)$. Graphically,
\[
\xymatrix{
\Sigma\ar[d]_{\varphi} & \ModI(\Sigma) \ar@{-}[rr]^{\models^\I_\Sigma}
& &\SenI(\Sigma)\ar[d]^{\SenI(\varphi)} \\
\Sigma' & \ModI(\Sigma')\ar[u]^{\ModI(\varphi)}\ar@{-}[rr]_{\models^\I_{\Sigma'}}
&&\SenI(\Sigma')\\
}
\]
\end{definition}

\begin{example}[Propositional Logic]\label{hpl-ex}
A signature $Prop \in |\Sign^{\PL}|$ is a set of
propositional variables symbols and a signature morphism is just a function
$\varphi:Prop\rightarrow Prop'$
Therefore, $\Sign^{\PL}$ coincides with the category $\Set$. 

Functor $\Mod$ maps each signature
$Prop$ to the  category $\Mod^{\PL}(Prop)$ and each signature
morphism $\varphi$ to the reduct functor
$\Mod^{\PL}(\varphi)$. Objects of $\Mod^{\PL}(Prop)$ are functions $M:Prop\rightarrow \{\top,
\bot\}$ and, its morphisms, functions $h:Prop\rightarrow Prop$
such that $M(p) = M'(h(p))$. Given a signature morphism
$\varphi:Prop\rightarrow Prop'$, the reduct of a model $M'\in
|\Mod^{\PL}(Prop')|$, say $M=\Mod^{\PL}(\varphi)(M')$ is defined, for
each $p\in Prop$, as $M(p)=M'(\varphi(p))$.

The sentences functor maps each
signature $Prop$ to the set of propositional sentences
$\Sen^{\PL}(Prop)$ and each morphism $\varphi:Prop\rightarrow Prop'$
to the sentences' translation
$\Sen^{\PL}(\varphi):\Sen^{\PL}(Prop)\rightarrow
\Sen^{\PL}(Prop')$. The set $\Sen^{\PL}(Prop)$ is the usual set
of propositional formulae defined by the grammar 
\[\rho ::= p \;|\; \rho\vee \rho \;|\; \rho \wedge \rho \;|\; \rho \Rightarrow
\rho \;|\; \neg \rho
\]
for $p\in Prop$. The translation of a sentence 
$\Sen^{\PL}(\varphi)(\rho)$ is obtained by replacing each proposition of
$\rho$ by the respective $\varphi$-image.

Finally, for each $Prop\in \Sen^{\PL}$, the satisfaction relation
$\models^\PL_{Prop}$ is defined as usual:
\begin{itemize}
  \item[--] $M\models^\PL_{Prop} p\; $ iff $\; M(p)=\top$, for any $p\in Prop$;
  \item[--] $M\models^\PL_{Prop} \rho \vee \rho'\; $ iff $\; M\models^\PL_{Prop} \rho$ or
    $M\models^\PL_{Prop} \rho'$,
\end{itemize}
and similarly for the other connectives.
\end{example}

\begin{example}[Equational logic]

\medskip
Signatures in the institution $\EQ$ of equational logic are  pairs $(S,F)$  where
 $S$ is a set of sort symbols and  $F = \{ F_{\arity\ra s} \mid \arity\in
 S^*, s\in S \}$ is a family of sets of 
operation symbols indexed by arities $\arity$ (for the arguments) and
sorts $s$ (for the results).
\emph{Signature morphisms} map both components in a compatible way:
they consist of pairs $\varphi = (\varphi^\st,  \varphi^\op) \co (S,F) \ra (S',F')$, where
$\varphi^\st \co S \ra S'$ is a function, and $\varphi^\op = 
\{ \varphi^\op_{\arity\ra s} \co F_{\arity\ra s} \ra F'_{\varphi^\st (\arity) \ra \varphi^\st (s)} 
\mid \arity\in S^*, s\in S \}$ a  family of functions mapping
operations symbols respecting arities. 

 A model $M$ for a
signature $(S,F)$ is an algebra
interpreting each sort symbol $s$ as a carrier set $M_s$ and each operation symbol 
$\sigma \in F_{\arity}\rightarrow s$ as a function
$M_\sigma:M_{\arity}\rightarrow M_s$, where $M_{\arity}$ is the product of the arguments' carriers.
Model morphism are 
homomorphisms of algebras, i.e., $S$-indexed families of functions
 $\{ h_s \co M_s \ra M'_s\mid s\in S \}$  such that for any $m\in
 M_{\arity}$, and for each  $\sigma \in F_{\arity \ra s}$,   $h_s (M_\sigma
 (m)) = M'_\sigma (h_{\arity} (m))$.
For each signature morphism $\varphi$, the \emph{reduct} 
of a model $M'$, say $M=\Mod^{\EQ}(\varphi)(M')$ is defined by $(M)_x =
M'_{\varphi(x)}$ for each sort and function symbol $x$ from
the domain signature of $\varphi$.
The models functor maps
signatures to categories of algebras and signature
morphisms to the respective reduct functors.

Sentences are
universal quantified equations $(\forall X) t=t'$.
Sentence translations along a signature morphism $\varphi:(S,F)
\rightarrow ( S',F')$, i.e.,  $\Sen^{\EQ}(\varphi):\Sen^{\EQ}(
S,F)\rightarrow \Sen^{\EQ}( S',F')$, replace symbols of $( S,F)$
by the respective  $\varphi$-images in $( S',F')$. The sentences functor maps each
signature to the set of first-order sentences and each signature
morphism to the respective sentences translation.
The satisfaction relation is the usual Tarskian
satisfaction defined recursively on the structure of the sentences
as follows:

\begin{itemize}\itemsep -1pt
\item $M \models_{(S,F)} t=t'$ when $M_t = M_{t'}$, where $M_t$ 
denotes the interpretation of the $(S,F)$-term $t$ in $M$ defined
recursively by 
$M_{\sigma(t_1,\dots,t_n)} = M_\sigma (M_{t_1},\dots,M_{t_n})$. 
\item $M \models_{(S,F)} (\forall X)\rho$ when 
$M' \models_{(S,F+X)} \rho$ for any $(S,F+X)$-expansion $M'$
of $M$.
\end{itemize}
\end{example}

\begin{example}[Propositional Fuzzy Logic]

Multi-valued logics  \cite{Gottwald01atreatise} generalise classic logics 
by replacing, as its \emph{truth domain}, the 2-element  Boolean algebra, by larger sets structured as
\emph{complete residuate lattices}.
They were originally
formalised as institutions in \cite{Agusti-CullellEGG90} (but see also  \cite{DBLP:journals/mlq/Diaconescu11} for a recent reference).

\emph{Residuate lattices }
are tuples $L=(\textbf{L},\leq, 
\wedge,\vee,\top,\bot,\otimes)$, where
\begin{itemize}
  \item $(\textbf{L},\wedge,\vee,\top,\bot)$ is a lattice ordered by $\leq$,
    with carrier $\textbf{L}$, with (binary) infimum ($\wedge$) and supremum ( $\vee$), 
    and bigest and
    smallest elements $\top$ and $\bot$;
  \item $\otimes$ is an associative binary operation such for any elements
    $x,y,z\in L$:
    \begin{itemize}
      \item $x \otimes \top= \top \otimes x= x$;
      \item  $y\leq z$ implies that $(x \otimes y)\leq (x \otimes z)$;
      \item there exists an element $x\Rightarrow z$ such that 
        \[y\leq (x\Rightarrow z) \text{ iff } x\otimes y\leq z.\]
    \end{itemize}
The residuate lattice $L$ is complete if any subset $ S\subseteq\textbf{L}$
has infimum and supremum denoted by $\bigwedge S$ and $\bigvee S$, respectively.
\end{itemize}

Given a complete residuate lattice $L$, the institution
$\MVL_L$ is defined as follows.
\begin{itemize}
   \item $\MVL_L$-signature are $\PL$-signatures.
  \item Sentences of $\MVL_L$ consist of pairs $(\rho,p)$ where $p$ is
    an element of $L$ and $\rho$ is defined as a $\PL$-sentence over the set of connectives $\{\Rightarrow
\vee,\top,\bot,\otimes\}$.

 \item A  $\MVL_L$-model $M$ is a function $M:FProp\rightarrow L$.
\item For any $M\in \Mod^{MVL_L}(FProp)$ and for any $(\rho,p)\in
  \Sen^{MVL_L}(FProp)$ the satisfaction relation is 
  \[M\models^{\MVL_L}_{FProp}(\rho,p)\; \text{ iff }\;  p\leq (M\models \rho)\]
where $M\models \rho$ is inductively defined as follows:
\begin{itemize}
  \item for any proposition $p\in FProp$,
    $(M\models p)=M(p))$; 
  \item $(M\models \top)=\top $; 
  \item $(M\models \bot)=\bot $; 
  \item $(M\models \rho_1 \star \rho_2)= (M\models \rho_1) \star
    (M\models \rho_2)$, for $\star\in\{ \vee,\Rightarrow,\otimes\}$;  
 \end{itemize}
\end{itemize}

This institution captures many multi-valued logics in
the literature. 
 For instance, taking $L$ as the {\L}ukasiewicz arithmetic lattice over
 the closed interval $[0,1]$, where $x\otimes y = 1- max\{0, x + y
 -1)\}$ (and $x \Rightarrow y = min\{1,1-x+y\}$), 
yields the standard \emph{propositional fuzzy logic}.
\end{example}

\subsection{Brief overview on the hybridisation method}

Having recalled the notion of an institution, we shall now briefly review the core of the \emph{hybridisation} method mentioned in the introduction and proposed in \cite{calco,paperdiaco}.
We concentrate in a simplified version, \ie, quantifier-free and non-constrained, of the general method.
The method enriches a base (arbitrary)
institution $\mathcal{I}= 
(\SignI, \SenI, \ModI, (\models^{\mathcal{I}}_{\Sigma})_{\Sigma \in
  |\SignI|})$ with  hybrid logic
features and the corresponding Kripke semantics. The result is  still an institution, $\HI$, 
called the \emph{hybridisation of $\I$}. 

\medskip
\noindent \emph{The category of $\HI$-signatures.} First of all
the base signature is enriched with nominals and 
polyadic modalities. Therefore, the category of \emph{$\mathcal{I}$-hybrid signatures}, denoted by  $\SIGN$,
is defined as the direct (cartesian) product of categories:
\[
\SIGN=\SignI \times \Sign^{\REL}. 
\]
Thus, signatures are triples $(\Sigma, \Nom, \Lambda)$, where
$\Sigma\in |\SignI|$ and, in the $\REL$-signature $(\Nom, \Lambda)$, $\Nom$ 
is a set of constants called \emph{nominals} and $\Lambda$ 
is a set of relational symbols called \emph{modalities}; 
$\Lambda_n$ stands for the set of modalities of arity $n$. Morphisms
$\varphi\in \SIGN((\Sigma,\Nom,\Lambda),(\Sigma',\Nom',\Lambda'))$
are  triples $\varphi=(\varphi_\Sig,\varphi_\Nom,\varphi_\MS)$
where $\varphi_\Sig\in \SignI(\Sigma,\Sigma')$,
$\varphi_\Nom:\Nom\rightarrow\Nom'$ is a function and
$\varphi_\MS=(\varphi_n:\Lambda_n\rightarrow \Lambda'_n)_{n\in
\mathbb{N}}$ a $\mathbb{N}$-family of functions mapping nominals and
$n-ary$-modality symbols, respectively.

\medskip
\noindent\emph{$\mathcal{HI}$-sentences functor.} The second step is  to enrich the
base sentences accordingly. The sentences of the base institution and the nominals are taken as atoms
and  composed with the boolean connectives,
modalities, and satisfaction operators as follows:
$\SEN(\Sigma, \Nom,\Lambda)$ is the least set such that 
\begin{itemize}\itemsep -1pt

\item $\Nom \subseteq \SEN(\Delta)$;

\item $\Sen^{\I} (\Sigma)\subseteq \SEN(\Delta)$;

\item $\rho \star \rho' \in \SEN(\Delta)$ for any
  $\rho, \rho'\in \SEN(\Delta)$ and any
  $\star \in \{\vee,   \wedge,\implies\}$,  

\item  $\neg \rho \in \SEN(\Delta)$, for any $\rho\in \SEN(\Delta)$,

\item $@_i \rho\in \SEN(\Delta) $ for any $\rho \in \SEN(\Delta)$ and $i \in \Nom$;

\item $[\lambda](\rho_1,\dots,\rho_n), \langle \lambda \rangle(\rho_1,\dots,\rho_n)
\in \SEN(\Delta)$, for any $\lambda\in \Lambda_{n+1}, \rho_i \in \SEN(\Delta)$, 
$i\in \{1,\dots, n\}$.

\end{itemize}

\noindent
Given a $\mathcal{HI}$-signature morphism
$\varphi=(\varphi_{\Sig},\varphi_{\Nom},\varphi_{\mathrm{\MS}}) \co 
( \Sigma, \Nom, \Lambda )\ra ( \Sigma', \Nom', \Lambda' )$, 
the translation of sentences $\SEN(\varphi)$ is defined as follows:
\begin{itemize}\itemsep -1pt

\item $\SEN(\varphi)(\rho)=\Sen^\mathcal{I}(\varphi_{\Sig})(\rho)$ 
for any $\rho\in \Sen^\mathcal{I}(\Sigma)$;

\item $\SEN(\varphi)(i)=\varphi_{\Nom}(i)$;

\item $\SEN(\varphi)(\neg \rho)=\neg\SEN(\varphi)(\rho)$;

\item $\SEN(\varphi)(\rho \star \rho')=\SEN(\varphi)(\rho)\star\SEN(\varphi)(\rho')$, 
$\star \in \{\vee, \wedge, \implies\}$;

\item $\SEN(\varphi)(@_i \rho)=@_{\varphi_{\Nom}(i)}\SEN(\rho)$;

\item $\SEN(\varphi)([\lambda](\rho_1,\dots,\rho_n))=
[\varphi_{\mathrm{\MS}}(\lambda)](\SEN(\rho_1),\dots,\SEN(\rho_n))$;

\item $\SEN(\varphi)(\langle\lambda\rangle(\rho_1,\dots,\rho_n))=
\langle\varphi_{\mathrm{\MS}}(\lambda)\rangle(\SEN(\rho_1),\dots,\SEN(\rho_n))$.

\end{itemize}

\emph{$\mathcal{HI}$-models functor.} Models of the hybridised logic $\HI$ can be regarded
as ($\Lambda$-)Kripke structures whose worlds are 
$\I$-models. Formally \emph{$(\Sigma,\Nom,\Lambda)$-models} are pairs 
$(M, W)$ where 
\begin{itemize}\itemsep -1pt

\item $W$ is a $(\Nom,\Lambda)$-model in $\REL$;

\item $M$ is a function $|W| \rightarrow |\ModI (\Sigma)|$.

\end{itemize}
In each  world  $(M,W)$,  
$\{ W_n \mid n \in \Nom \}$ provides interpretations for
\emph{nominals} in $\Nom$, whereas relations 
$\{ W_\lambda \mid \lambda\in \Lambda_n, n\in \omega \}$ 
interprete  \emph{modalities} $\Lambda$. 
We denote $M(\w)$ simply by $M_{\w}$.
The reduct definition is lifted from the base institution: the reduct
of a $\Delta'$-model $(M',W')$ along a signature morphism
$\varphi=(\varphi_{\Sig},\varphi_{\Nom},\varphi_{\mathrm{\MS}}):\Delta\rightarrow
\Delta'$, 
denoted by $\MOD(\varphi)(M',W')$, is the $\Delta$-model
$(M,W)$ such that
\begin{itemize}\itemsep -1pt

\item $W$ is the $(\varphi_{\Nom},\varphi_{\mathrm{\MS}})$-reduct of $W'$; i.e.

\begin{itemize}

\item $|W|=|W'|$;

\item for any $n \in \Nom, W_n =W'_{\varphi_{\Nom}(n)}$;

\item for any $\lambda \in \Lambda$,
  $W_\lambda =W'_{\varphi_{\mathrm{\MS}}(\lambda)}$; 

\end{itemize}

\item for any $\w\in |W|$, $M_\w=\Mod^\mathcal{I}(\varphi_{\Sig})(M'_\w).$

\end{itemize}

\medskip 
\noindent \emph{The Satisfaction Relation.} Let
$(\Sigma,\Nom,\Lambda)\in |\SIGN|$ and 
$(M,W)\in |\MOD(\Sigma, \Nom, \Lambda)|$. For
any $\w\in |W|$ we define:

\begin{itemize}\itemsep -1pt
		
\item $(M,W)\models^\w \rho$ iff $M_\w\models^\mathcal{I} \rho$; 
when $\rho \in \Sen^\mathcal{I}(\Sigma)$,

\item $(M,W)\models^\w i$ iff  $W_i=\w$; when $i\in \Nom$,
		
\item $(M,W)\models^\w \rho \vee \rho' $ iff  $(M,W)\models^\w \rho$  or  
$(M,W)\models^\w \rho'$,
		
\item $(M,W)\models^\w \rho \wedge \rho' $ iff  $(M,W)\models^\w \rho$  and  
$(M,W)\models^\w \rho'$,

\item $(M,W)\models^\w \rho \implies \rho' $ iff  
$(M,W)\models^\w \rho$  implies that  $(M,W)\models^\w \rho'$,

\item $(M,W)\models^\w \neg \rho$ iff  
$(M,W)\not{\models^\w} \rho$,

\item $(M,W)\models^\w [\lambda](\xi_1,\dots,\xi_n)$ iff for any 
$(\w,\w_1,\dots,\w_n) \in W_\lambda$ we have that
$(M,W)\models^{\w_i}\xi_i$ for some $1\leq i\leq n$.

\item $(M,W)\models^\w 
    \langle\lambda\rangle(\xi_1,\dots,\xi_n)$ iff there exists $(\w,\w_1,\dots,\w_n)
    \in W_\lambda$ such that and $(M,W)\models^{\w_i}\xi_i$ for any  
    $1\leq i\leq n$.

\item $(M,W)\models^\w @_j \rho$ iff $(M,W)\models^{W_j} \rho$,

\end{itemize}
We write $(M,W)\models \rho$ iff $(M,W) \models^\w \rho$ for any $\w \in |W|$.

\medskip \noindent As expected $\HI$ is itself an
institution:
\begin{theorem}[\cite{calco}]
Let $\Delta=(\Sigma, \Nom, \Lambda)$ and $\Delta'=(\Sigma',\Nom', \Lambda')$ be
two $\mathcal{HI}$-signatures and $\varphi\co\Delta\rightarrow \Delta'$ 
a morphism of signatures. For any $\rho \in \SEN(\Delta)$,  $(M',W')\in |\Mod^{C}(\Delta')|$,
and $\w\in |W|$,
\begin{center}
$\MOD(\varphi)(M',W') \models^\w \rho$ 
iff  
$(M',W')\models^\w \SEN(\varphi)(\rho).$
\end{center}
\end{theorem}

Let us illustrate the method by applying it to the three institutions described above.

\begin{example}[$\HPL$]\label{ex-hpl} 
The hybridisation of the propositional logic institution $\PL$ is an institution where
 signatures are triples $(Prop,\Nom,\Lambda)$ and sentences are generated by
\begin{equation}\label{sen-HPL}
\rho ::= \rho_0 \;|\; i \;|\; @_i \rho \:|\; \rho\odot \rho
\;|\; \neg \rho \;|\; \langle \lambda \rangle (\rho,\dots,\rho) \;|\; [\lambda] (\rho,\dots,\rho)
\end{equation}
where $\rho_0\in \Sen^{PL}(Prop)$, $i\in \Nom$, $\lambda \in \Lambda_n$
and $\odot=\{\vee,\wedge,\Rightarrow\}$. Note there is a double
level of connectives in the sentences: the one coming from base $\PL$-sentences
and another introduced by the hybridisation process. However, they 
 ``semantically collapse'' and, hence, no distinction between them needs to be done
(see \cite{paperdiaco} for details). 
A
$(Prop,\Nom,\Lambda)$-model is a pair $(M,W)$, where $W$ is a transition
structure with a set of worlds $|W|$. Constants $W_i, i\in
\Nom$ stand for the named worlds and  $(n+1)$-ary relations $W_\lambda$, $\lambda \in
\Lambda_n$ are the accessibility relations characterising the structure. For
each world $w\in |W|$, $M(w)$ is a (local) $\PL$-model, assigning 
propositions in $Prop$ to the world $w$. 

Restricting the signatures to those with just a single unary modality
(i.e., where $\Lambda_1=\{\lambda\}$ and $\Lambda_n=\emptyset$ for the
remaining $n\neq 1$), results in the usual institution for classical hybrid
propositional logic \cite{livro_brauner}.
\end{example}

\begin{example}[$\HMVL_L$]\label{ex-mvl}
The institution obtained through the hybridization of  $\MVL_L$, for a
fixed $L$, is similar to the $\HPL$ institution defined above, but
for two aspects,
\begin{itemize}
  \item sentences are defined as in 
    (\ref{sen-HPL}) but considering  $\MVL$
    $FPop$-sentences $(\rho_0,p)$ as atomic;
      \item to each world $w\in|W|$ is associated a function assigning to each proposition its  
  value in $L$.

\end{itemize}

It is interesting to note
that  expressivity increases even if one restricts to the case of  a (one-world) standard semantics.
For instance, differently from the base case where each sentence is tagged
by a $L$-value, one may now deal with more structured 
expressions involving several  $L$-values, as in, for example,  $(\rho,p)\wedge(\rho',p')$.
\end{example}

\begin{example}[$\HEQ$]\label{ex-heq}
Signatures of $\HEQ$ are triples $((S,F), 
\Nom,\Lambda)$ and the sentences are defined as in (\ref{sen-HPL}) but
taking $(S,F)$-equations $(\forall X) t=t'$ as atomic base
sentences. Models are Kripke structures with a
(local)-$(S,F)$-algebra per world. This institution is a suitable
framework to specify reconfigurable system in a
``configurations-as-worlds'' perspective: distinct
configurations are modelled by distinct algebras; and 
reconfigurations expressed by transitions
(c.f. \cite{sefm11,madeirathesis}). Clearly, in this sort of specifications  interfaces are given equationally, based on 
 $\EQ$-signatures. Nominals identify the ``relevant'' configurations and reconfigurations amount to state transitions.
 Therefore, one resorts to equations tagged with the satisfaction operators to specify
the configurations, plain equations to specify global properties of
the system and the modal features to specify its reconfigurability
dynamics.  
\end{example}

\section{Bisimulation for hybridised Logics}\label{sec-bis}
Having briefly reviewed what an institution is and how, through a systematic process, one may introduce in an arbitrary logic both modalities and nominals to explicitly refer to states in a specification, we may now focus on the paper's specific contribution. Our starting point is a method to specify reconfigurable software as transition systems whose states represent  particular configurations. They can themselves be an algebraic specification, a relation structure or even another, local transition system. Such two-staged specifications are common in the Software Engineering practice (see, e.g., Gurevich's Abstract State Machines  \cite{ASM03}); the originality of our method lies in its genericity: whatever logic is found useful to specify each concrete configuration, a method is offered to compute its hybrid counterpart. In this setting, this section and the following one seek for suitable notions of equivalence and refinement for this kind of specifications. Naturally, such notions should also be parametric on the base logic used, \ie, on the language in which the specifications of each concrete configuration are written. The price to pay is, of course, some extra notation and the use of a generic framework --- that of \emph{institutions} --- in which concepts can be formulated and results proved once and for all. 

As the external layer of a reconfigurable system specification is that of a transition system, it is natural to resort to suitable formulations of \emph{bisimilarity} and \emph{similarity} to capture equivalence and refinement, respectively. The precise characterisation of such notions at the high level of abstraction chosen, is, in fact, the paper's contribution.

Intuitively a bisimulation relates worlds which exhibit the ``same" (observable) information and preserves this property along transitions.
Thus, to  define a general notion of bisimulation over Kripke structures whose states are models of whatever base logic was chosen for specifications,
we have to make precise what  the ``same"  information actually means. For example, if the system's configurations are specified by \emph{equations}, as abstract data types, to establish that two such configurations are  bisimilar will certainly require that each specification generates the same variety. Actually, in this case, they are essentially the same data type.
In the more general
setting of this paper the base logic is a parameter and we have to deal with its hybridised version $\HI$. Our proposal is, thus, to resort to the broad notion of 
 \emph{elementary
 equivalence} (e.g.\cite{Hodges:1997:SMT:262326}), and add to the bisimulation definition the requirement that local configurations, \ie, local $\I$-models related by a bisimulation be
\emph{elementarily equivalent}. Formally, 

\begin{definition}
Let  $M, M'\in \ModI(\Sigma)$ and $\Sen'$ be a subfunctor of $\SenI$. Models $M$ and $M'$ are
elementarily equivalent with respect to sentences in $\Sen'(\Sigma)$, in symbols $M\equiv^{\Sen'} M'$, if for any $\rho\in
\Sen'(\Sigma)$
     \begin{equation} \label{elqsimple}
       M\models^\I \rho \; \, \text{iff}\; \, M'\models^\I \rho.
     \end{equation}
\end{definition}

Under the institution theory \emph{motto} --- \emph{truth is invariant under change of notation} ---  we write 
 $M\equiv^{\Sen'}_\varphi M'$ whenever $M\equiv^{\Sen'} \ModI(\varphi)(M')$ for 
a given $\varphi\in \SignI(\Sigma, \Sigma')$, $M\in \ModI(\Sigma)$ and
     $M'\in \ModI(\Sigma')$.    Models $M$ and $M'$ are said to be  $\varphi,\Sen'$-elementarily equivalent.
     
Resorting to the satisfaction condition in $\I$, the following characterisation of  $\varphi,\Sen'$-elementary equivalence pops out:

\begin{corollary}\label{corelq}  $M\equiv^{\Sen'}_\varphi M'\; $ iff, for any $\, \rho \in \Sen'(\Sigma)\text{, }\; M\models^{\I}_\Sigma \rho
    \Leftrightarrow  M'\models^{\I}_\Sigma \SignI(\varphi)(\rho)$.
\end{corollary}
\noindent
If only an implication $\Rightarrow$ holds in the right hand side of the above equivalence we write $M\feq^{\Sen'} M'$.
Note the  role of $\varphi$ above: as a signature morphism it captures the possible \emph{change of notation} from a specification to another.
For example it may cater for renaming 
propositions in Ex.~\ref{ex:HPL} or  signature components in
Ex.~\ref{ex:HEQ}. However, its pertinence becomes clearer in 
refinement situations, as discussed in the next section. There it may accommodate many forms of interface
enrichment or adaptation
(e.g. through the introduction of auxilliar operations).

Let us now define bisimulation in this general setting.

\begin{definition}\label{generalbisimulation}
Let $\HI$ be the hybridization of the institution $\I$ and $\varphi\in
\SIGN(\Delta,\Delta')$ a signature morphism. Let $\Sen'$ be a subfunctor of $\SenI$. A
\emph{$\varphi,\Sen'$-bisimulation between  models $(M,W)\in
  \MOD(\Delta)$ and 
$(M',W')\in \MOD(\Delta')$ } is a non-empty relation $\bis
\subseteq |W|\times |W'|$
such that

\begin{enumerate}
 \item[(i)]\label{defI}  for any $w \bis w'$, and for any $i\in \Nom$, 
$W_i=w  \text{ iff } W'_{\varphi_\Nom(i)}=w'$.
  \item[(ii)]\label{defII} for any $w \bis  w'$, 
    $M_w\equiv^{\Sen'}_{\varphi_\Sig}M'_{w'}$.
    \item[(iii)]\label{defIII} 
    for any $i\in \Nom$, $W_i \bis W'_{\varphi_\Nom(i)}$.
  \item[(iv)]\label{defIV} For any $\lambda \in
      \Lambda_n$, if $(w,w_1,\dots,w_n)\in W_\lambda$ and $w \bis w'$, then for each  $k\in \{1,\dots,n\}$ there is a
      $w'_k\in |W'|$ such that $w_k \bis w'_k$ 
      and $(w',w'_1,\dots, w'_n)\in  W'_{\varphi_\MS(\lambda)}$. 
  \item[(v)]\label{defV} For any $\lambda \in
      \Lambda_n$ if $(w',w'_1,\dots,w'_n)\in
      W'_{\varphi_\MS(\lambda)}$ and $w \bis w'$, then for each  $k\in \{1,\dots,n\}$ there is a
      $w_k\in |W|$, such that $w_k \bis w'_k$ and $(w,w_1,\dots,
      w_n)\in  W_\lambda$. 
\end{enumerate}
\end{definition}

The following result  establishes that, for quantifier-free hybridisations,
the (local)-hybrid satisfaction $\models^\HI$ is 
invariant under $\varphi,\Sen$-bisimulations:

\begin{theorem}\label{bisinvariance}

Let $\HI$ be a quantifier-free hybridization of the institution $\I$ and $\varphi\in
\SIGN(\Delta,\Delta')$ a signature morphism. Let $\bis
\subseteq |W|\times |W'|$ be a $\varphi,\Sen$-bisimulation.
 Then, for any $w \bis w'$
  and for any $\rho \in \SEN(\Delta)$,
\begin{equation}
  (M,W)\models^w \rho \text{ iff }   (M',W')\models^{w'} \SEN(\varphi)(\rho).
\end{equation}
\end{theorem}

  \begin{proof}
The proof is by  induction on the structure of the sentences.

\begin{enumerate}

\item $\rho=i$ for some $i\in \Nom$: 
	\begin{eqnarray*}
	& & (M, W)\models^\w i
	\just\Leftrightarrow{ defn. of $\models^w$}
        W_i =\w	
        \just\Leftrightarrow{ (i) of Defn \ref{generalbisimulation} }
        W'_{\varphi(i)} =\w'	
        \just\Leftrightarrow{ defn. of $\models^{w'}$}
        (M',W')\models^{\w'}\varphi_{\Nom}(i)
       \just\Leftrightarrow{ defn of $\SEN(\varphi)$}
       (M',W')\models^{\w'}\SEN(\varphi)(i)
\end{eqnarray*}

\item $\rho \in \Sen^{\mathcal{I}}(\Sigma)$:
	\begin{eqnarray*}
	& & (M, W)\models^\w \rho
	\just\Leftrightarrow{ defn. of $\models^w$}
        M_\w \models^{\I} \rho
        \just\Leftrightarrow{by hypothesis
          $M_w\equiv_{\varphi_\Sig}M'_{w'}$ + Cor~\ref{corelq}}
        M'_{\w'}\models \SenI(\varphi_{\Sig})(\rho)
	\just\Leftrightarrow{ defn. of $\models^{w'}$}
        (M',W')\models^{\w'}  \SenI(\varphi_{\Sig})(\rho)
	\just\Leftrightarrow{ defn of $ \SEN(\varphi)$}
        (M',W')\models^{\w'} \SEN(\varphi)(\rho)
\end{eqnarray*}

\item $\rho=\xi \vee \xi'$ for some $\xi, \xi'\in \SEN(\Delta)$:
	\begin{eqnarray*}
	& & (M, W)\models^\w \xi \vee \xi'
	\just\Leftrightarrow{ defn. of $\models^w$}
        (M, W)\models^\w \xi\; \text{or}\; (M, W)\models^\w \xi' 
        \just\Leftrightarrow{ I.H.}
       (M',W')\models^{\w'} \SEN(\varphi)(\xi)\; \text{or}\;
       \\ & & (M',W')\models^{\w'} \SEN(\varphi)(\xi') 
	\just\Leftrightarrow{ defn. of $\models^w$}
        (M',W') \models^{\w'} \SEN(\varphi)(\xi \vee \xi')
        \end{eqnarray*}

\noindent
The proofs for cases  $\rho=\xi \wedge \xi'$, $\rho=\xi\implies \xi'$, 
$\rho=\neg\xi$, etc. are analogous.

\item $\rho = [\lambda](\xi_1,\dots,\xi_n)$ for some 
$\xi_1,\dots, \xi_n \in \SEN(\Delta)$, $\lambda\in\Lambda_{n+1}$:

	\begin{eqnarray*}
	& & (M, W)\models^\w [\lambda](\xi_1~,\dots, \xi_n)
	\just\Leftrightarrow{ defn. of $\models^w$}
        \text{for any}\;  (\w,\w_1,\dots,\w_n)\in W_\lambda  \text{
          there is some } k\in \{1,\dots,n\} \\ & & \text{ such that }
        (M,W)\models^{\w_k}\xi_k 
        \just\Leftrightarrow{\textbf{*} }
         \text{for any}\;  (\w',\w'_1,\dots,\w'_n)\in
         W'_{\varphi_{\MS}(\lambda)}   \text{ there is some }  \\ 	& & p\in \{1,\dots,n\}  \text{
           such that } 	  (M',W')\models^{\w'_p}\SEN(\varphi)(\xi_p) 
        \just\Leftrightarrow{ defn. of $\models^{w'}$}
        (M',W')\models^{\w'}
        [\varphi_{\MS}(\lambda)](\SEN(\varphi)(\xi_1),\dots,
        \SEN(\varphi)(\xi_n))  
        \just\Leftrightarrow{ defn. of $\SEN(\varphi)$}
         (M',W')\models^{w'} \SEN(\varphi)([\lambda](\xi_1,\dots, \xi_n)) 
        \end{eqnarray*}
\noindent For the step marked with \textbf{*} we proceed as follows.  Supposing $(w',w'_1,\dots,w'_n)\in W'_{\varphi_\MS(\lambda)}$
with $w\bis w'$, we have by clause (v) of Defn.
\ref{generalbisimulation} that there are $w_k$, with $k\in\{1,\dots, n\}$, such
that $(w,w_1,\dots,w_n)\in W_{\lambda}$. By hypothesis,  $(M,W) \models^{w_p} \xi_p$ for some
$p\in\{1,\dots,n\}$. Moreover, by
I.H. $(M',W')\models^{\w'_p}\SEN(\varphi)(\xi_p)$. Clause (iv) of Defn.
\ref{generalbisimulation} entails the converse implication.
The proof for sentences of form $\rho=\langle\lambda\rangle(\xi_1,\dots,\xi_n)$
is analogous. 
		
\item$\rho=@_i\xi$ for some $\xi \in \SEN(\Delta)$ and $i \in \Nom$:

	\begin{eqnarray*}
	& & (M,W)\models^\w @_i\xi
	\just\Leftrightarrow{ defn. of $\models^w$}
        (M,W)\models^{W_i}\xi
        \just\Leftrightarrow{I.H. and clause (iii) of Defn~\ref{generalbisimulation}}
        (M',W')\models^{W'_{\varphi_\Nom(i)}}\SEN(\varphi)(\xi)
        \just\Leftrightarrow{ defn. of $\models^w$}
        (M',W')\models^{\w} @_{\varphi_\Nom(i)}\SEN(\varphi)(\xi)
        \just\Leftrightarrow{ defn. of $\SEN(\varphi)$}
        (M',W')\models^{\w}\SEN(\varphi)(@_i\xi)
      \end{eqnarray*}

\end{enumerate}

\end{proof}

As direct consequence of the previous theorem we get the following
characterisation of the preservation of (global) satisfaction, $\models^{\HI}$,
under $\varphi$-bisimilarity:

\begin{corollary}\label{bisinvariance1}
On the conditions of Theorem~\ref{bisinvariance},  let $(M,W) \bisim(M',W')$ witnessed by a total and surjective  bisimulation. Then,
\begin{equation}
  (M,W)\models^{\HI} \rho\;  \text{ iff }\;    (M',W')\models^{\HI}\SEN(\varphi)(\rho).
\end{equation}
\end{corollary}

\begin{example}[Bisimulation in $\HPL$]\label{ex:HPL}
Let us instantiate Defn. \ref{generalbisimulation} for the $\HPL$ case
(cf. Ex. \ref{hpl-ex}), considering $\varphi= id$ and $\Sen'=\SenI$. A bisimulation $\mathrm{B}$ is such that 
 $(M,W)\mathrm{B}(M',W')$,  for any two models $(M,W),(M',W')\in
|\Mod^{\HPL}(P,\Nom,\{\lambda\})|$, if
\begin{itemize}
 \item[(i)] for any $i\in \Nom$, $w\mathrm{B} w'$, $w=W_i \text{ iff }w'=W'_i$;
 \item[(ii)] $M_w\equiv M'_{w'}$, i.e., bisimilar states satisfy
    the same sentences;
 \item[(iii)] for any $i\in \Nom$, $W_i \mathrm{B} W'_i$;
  \item[(vi)] for any $(w,w_1)\in W_\lambda$ with $w\mathrm{B}  w'$, there is
    a  $w'_1\in |W'|$ such that $w_1 \mathrm{B} w'_1$ 
      and $(w_1,w'_1)\in  W'_{\lambda}$;
  \item[(v)] for any $(w',w'_1)\in W'_\lambda$ with $w\mathrm{B}
       w'$, there is a
      $w_1\in |W|$ such that $w_1\mathrm{B} w'_1$ 
      and $(w_1,w'_1)\in  W'_{\lambda}$;. 
\end{itemize} 
Note that condition (ii) is equivalent to say that bisimilar states 
have assigned the same set of propositions (for any $p\in P$,
    $M_w(p)=\top$ iff $M'_{w'}(p)=\top$).
As expected, this definition 
corresponds exactly to standard bisimulation for propositional hybrid logic
(see, e.g. \cite[Defn 4.1.1]{vbt}).
\end{example}

The definition of bisimulation computed in the previous example, can also capture the case of propositional
modal logic: just consider pure modal
signatures (i.e., with an empty set of nominals), as condition $(i)$ is trivially
satisfied. Moreover, instantiating Theorem \ref{bisinvariance} we get the
classical result about preservation of modal truth by  bisimulation.

\begin{example}[Bisimulation for $\HEQ$]\label{ex:HEQ}
Consider now  the instantiation of \ref{generalbisimulation} for $\HEQ$
(cf. Ex~\ref{ex-heq}). All one has to do is to  replace
condition (iv) in Defn~\ref{generalbisimulation} by its instantiation for algebras:
 two algebras are
elementarily equivalent if the respective generated varieties
coincides \cite{gratzer}. 
\end{example}

\section{Refinements for generic hybridised logics}\label{sec-sim}

Let us come back to the general case of a reconfigurable system described by a set of configurations and a transition structure entailing changes from one to another. If equivalence of specifications of such systems corresponds to a notion of bisimilarity in which bisimilar configurations are enforced to be elementary equivalent, a \emph{refinement} relation corresponds to \emph{similarity}. This entails, on the one hand,  preservation (but not reflection ) of transitions, \ie, of reconfiguration steps, from the abstract to the concrete system. And, on the other hand, at each local configuration, preservation of the original properties along local refinement. Formally,

\begin{definition}\label{generalsimulation}
Let $\HI$ be the hybridisation of an institution $\I$, $\varphi\in
\SIGN(\Delta,\Delta')$ a signature morphism and $\Sen'$ a subfunctor of $\SenI$.
A \emph{$\varphi,\Sen'$-refinement of $(M,W)\in
  \MOD(\Delta)$ by $(M'W')\in \MOD(\Delta')$ } consists of a
non-emtpy relation $\simf^{\Sen'}
\subseteq |W|\times |W'|$ such that, for any $w\simf^{\Sen'} w'$,

\begin{enumerate}
  \item[(f.i)]\label{i-1} for any $i\in \Nom$, $\text{if }W_i=w  \text{ then } W'_{\varphi_\Nom(i)}=w'$.
  \item[(f.ii)]\label{i-2} $M_w \feq^{\Sen'} M'_{w'}$.
    \item[(f.iii)]
    for any $i\in \Nom$, $W_i\, \simf^{\Sen'}\,  W'_{\varphi_\Nom(i)}$.
  \item[(f.iv)]\label{i-3} For any $\lambda \in
      \Lambda_n$, if $(w,w_1,\dots,w_n)\in W_\lambda$ then for each  $k\in \{1,\dots,n\}$ there is a
      $w'_k\in |W'|$ such that $w_k \simf  w'_k$ 
      and $(w',w'_1,\dots, w'_n)\in  W'_{\varphi_\MS(\lambda)}$. 
\end{enumerate}

\end{definition}

The question is, now, to see whether (hybrid)
satisfaction is, or is not, preserved by refinement. On a first attempt, it is natural to accept a positive answer which, although intuitive, is wrong.
 Actually, 
not  all  hybrid sentences can be preserved along a refinement chain. Note on the
proof of Th~\ref{bisinvariance}, that the preservation of  hybrid satisfaction
of sentences $[\lambda](\xi_1,\dots,\xi_n)$ is entailed by 
condition  $(ii)$ of Defn~\ref{generalbisimulation}, but the latter is stated on the opposite direction to  refinement.
As a simple counter-example,
 define a $\simf$-refinement  from a
$\Delta$-hybrid model $(M,W)$ with 
$|W|=\{w\}$ and $W_\lambda=\emptyset$ for $\lambda \in \Lambda_n$ to any other
$\Delta'$-hybrid model $(M',W')$ such that
$\Mod^\HI(\varphi_\Sig)(M'_{w'})=M_w$ for some $w'\in |W'|$. Sentence
$[\lambda](\xi_1,\dots,\xi_n)$, which trivially holds in the world $w$ of
$(M,W)$, may fail to
be satisfied in the $\simf$-related world $w'$ of
$(M',W')$.  Sentences like $\neg \xi$ provide another counter-example. The reason is that, by
hypothesis,  preservation is only assumed on the refinement direction and, of course, non satisfaction in one direction, does not implies 
non satisfaction in the other. Therefore, differently from  the
bisimulations case, the preservation of the satisfaction under
refinement  does not hold
for all the hybrid sentences universe.  Actually, the `boxed'
and negated sentences are exactly the cases where it may fail. 

Finally, a note regarding parameter $\Sen'$ in condition \emph{(f.ii)}. 
First of all  note that the ``unrestricted'' implication of clause
\emph{(f.ii)} in Defn~\ref{generalsimulation} 
is very strong: it often implies the converse
implication as well. For instance, in $\HPL$, the condition holds iff
$M_w=\Mod(\varphi)(M'_{w'})$.  In particular, an $id$-refinement
implies the equality of  realizations of related worlds (since,
the implication ``$M_w\models^\PL\neg p \text{ then
}M'_{w'}\models^\PL\neg p$'' 
is equivalent to the implication ``$M'_{w'}\models^\PL p \text{ then }
M_w\models^\PL p$''. 
Hence, $M_w=M'_{w'}$). 
It seems reasonable to  weaken this condition to yield a  strict inclusion. One way to do this is to restrict the focus to a subset of the sentences in the base institution. In the 
example mentioned above this will correspond to exclude  $\PL$
negations, which amounts to take as  $\Sen' (Prop)$ the set of
propositional sentences without negations.

Given an institution
$\I=(\SignI,\SenI,\ModI,(\models_\Sigma)_{\Sigma\in |\SignI|})$ and a
sentences subfunctor $Sen'\subseteq \SenI$,  we denote by
$\HI'$ the hybridisation of the institution
$\I'=(\SignI,\Sen',\ModI,(\models_\Sigma)_{\Sigma\in |\SignI|})$.

\begin{definition}[$\Sen'$-Positive Existencial sentences]
The $\Sen'$-positive existencial sentences of a signature $\Delta\in |\SIGN|$
are given by a subfunctor $\SEN_+\subseteq \Sen^{\HI'}$ defined
inductively for each signature $\Delta$  as $\Sen^{\HI'}(\Delta)$ but excluding
both negations and box modalities. For each signature morphism
$\varphi:\Delta\rightarrow \Delta'$, $\SEN_+(\varphi)$ is the
restriction of $\Sen^{\HI'}(\varphi)$ to $\SEN_+(\Delta)$.
\end{definition}

\begin{theorem}\label{siminvariance}
Let $\HI$ be the quantifier free hybridisation of an institution
$\I$,  $\Sen'$ a subfunctor of  $\SenI$, $\varphi\in
\SIGN(\Delta,\Delta')$ a signature morphism,
$\simf^{\Sen'}$ a $\varphi,\Sen'$-refinement relation  and $(M,W)\in
  \MOD(\Delta)$ and $(M',W')\in \MOD(\Delta')$ two mo\-dels such that $(M',W')$ is a refinement of $(M,W)$ witnessed by relation $\simf^{\Sen'}$.
 Then, for any $w \simf^{\Sen'} w'$  and  $\rho \in \SEN_+(\Delta)$,
\[
  (M,W)\models^w \rho \text{ implies that }   (M',W')\models^{w'} \SEN(\varphi)(\rho).
\]
\end{theorem}
\begin{proof}
The proof is  by induction on the structure of the existential
positive sentences and comes directly from the proof of
Th~\ref{bisinvariance}, taking the right to left implication.
 Preservation of base sentences follows exactly the
same proof since the IH is precisely about the $\Sen'$ sentences.
 What remains to be proved is the case
$\rho=\langle \lambda \rangle(\xi_1,\dots,\xi_n)$. Thus,
	\begin{eqnarray*}
	& & (M, W)\models^\w \langle \lambda\rangle (\xi_1~,\dots, \xi_n)
	\just\Leftrightarrow{ defn. of $\models^w$}
        \text{there exists}\;  (\w,\w_1,\dots,\w_n)\in W_\lambda \\ & & \text{ such that }
        (M,W)\models^{\w_k}\xi_k  \text{ for any } k\in \{1,\dots,n\} 
        \just\Rightarrow{By (f.iii), we have $w_k\simf
          w'_k$ for any $k\in \{1,\dots,n\} $ + I.H. }
        \text{there exists}\;  (\w',\w'_1,\dots,\w'_n)\in W'_{\varphi_{\MS}(\lambda)} \\ & & \text{ such that }
        (M',W')\models^{\w'_k}\xi_k  \text{ for any } k\in \{1,\dots,n\}         %
        \just\Leftrightarrow{ defn. of $\models^{w'}$}
        (M',W')\models^{\w'}
        \langle \varphi_{\MS}(\lambda)\rangle(\SEN(\varphi)(\xi_1),\dots,
        \SEN(\varphi)(\xi_n))  
        \just\Leftrightarrow{ defn. of $\SEN(\varphi)$}
         (M',W')\models^{w'} \SEN(\varphi)(\langle \lambda\rangle(\xi_1,\dots, \xi_n)) 
        \end{eqnarray*}
\end{proof}

\begin{corollary}\label{siminvariance1}
In the conditions of Th~\ref{siminvariance}, for any $\rho \in
\SEN_+(\Delta)$, if $\simf$ is surjective, then
    \[
  (M,W)\models \rho \text{ implies that }   (M',W')\models \SEN(\varphi)(\rho).
\]
\end{corollary}

\noindent
The following examples illustrate refinement situations in this setting.

\begin{example}[Refinement in $\HMVL_L$] 
Figure \ref{fig1} illustrates an example of a $\Sen'$-refinement in
$\HMVL_{L_4}$, for $L_4$ represented in Figure \ref{fig1}. Consider $\Sen'\subseteq \SenI$ restricting the base sentences 
to propositions, i.e.,  $\Sen'(LProp)=\{(p,l)|p\in LProp \text{ and }
l\in L_4\}.$
\begin{figure}[h]\label{fig1}
\begin{center}
\includegraphics[width=0.7\linewidth]{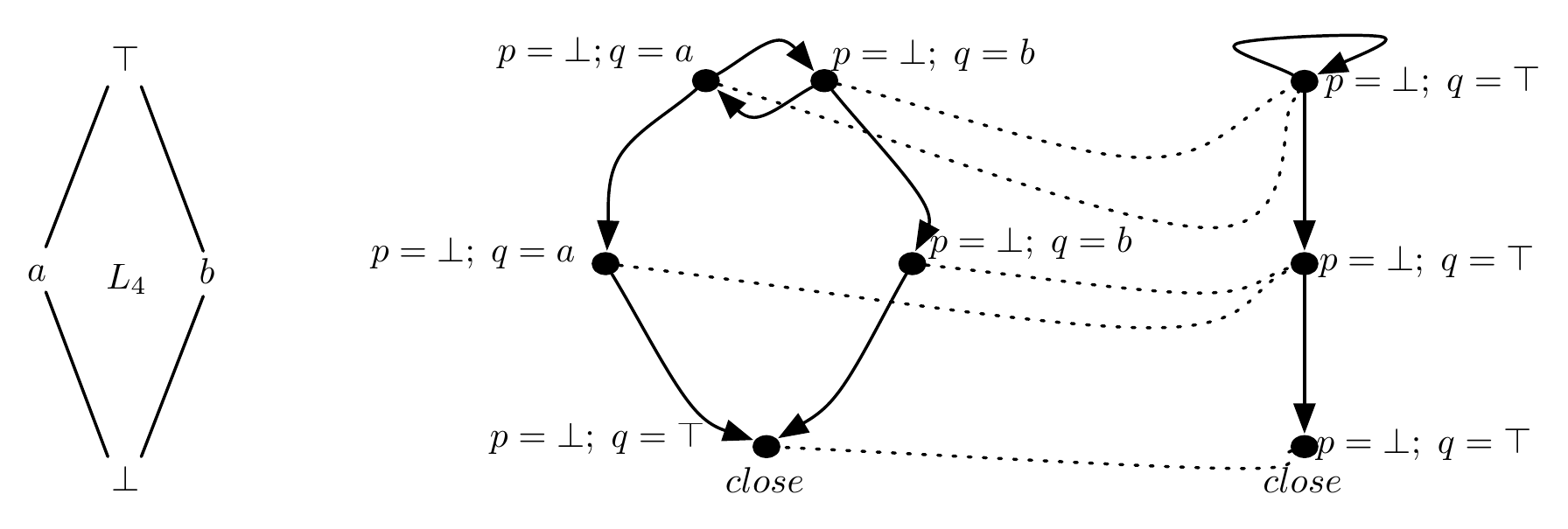}
\end{center}
\caption{Refinement  in $\HMVL_L$}
\end{figure}
Conditions (f.i) and (f.iii) are obviously satisfied. In what concerns the
verification of condition (f.ii) for which 
$(p,l)\in \Sen'(LProp)$, $M_w\models_{LProp}^{\MVL_{L_4}} (p,l)\Rightarrow M'_{w'}\models_{LProp}^{\MVL_{L_4}}
(p,l)$, it is sufficient to be that, $(M_w\models p) \leq
(M'_{w'}\models p)$, $p\in LProp$.

\end{example}
\begin{example}[Refinement in $\HEQ$]
Consider a store system abstractly modelled as the initial algebra $A$
of the $((S,F),\Gamma)$  where $S=\{mem,
elem\}$, $F_{mem\times elem\rightarrow mem}=\{write\}$,
$F_{mem\rightarrow mem}=\{del\}$ and $F_{\arity\rightarrow
  s}=\emptyset$ otherwise and  $\Gamma=\{del(write(m,e))=m\}$.
Suppose one intends to refine this structure into a $read$ function
configurable in two different modes: in one of them it reads the first element in the store, in the other the last.
 Reconfiguration between the two execution modes is enforced by an external event $shift$. Note that the
abstract model can be seen as the $\big((S,F),\emptyset,\{shift\}\big)$-hybrid model $\mathcal{M}=(M,W)$, taking $|W|=\{\star\}$,
$W_{shift}=\emptyset$ and $M_\star=A$. Then, we take the 
inclusion morphism $\varphi_\Sig:(S,F)\hookrightarrow (S,F')$ where
$F'$ extends $F$ with $F_{mem\rightarrow elem}={read}$ and
$F_{mem}=\{empty\}$. For the envisaged refinement let us consider the model
$\mathcal{M}'=(M',W')$ where $W'=\{s_1,s_2\}$ and
$W'_{shift}=\{(s_1,s_2),(s_2,s_1)\}$ and where $M_{s_1}$ and $M_{s_2}$
are the initial algebras of the equations presented in Figure
\ref{ex2}. 
\begin{figure}[h]\label{ex2}
\begin{center}
\includegraphics[width=0.7\linewidth]{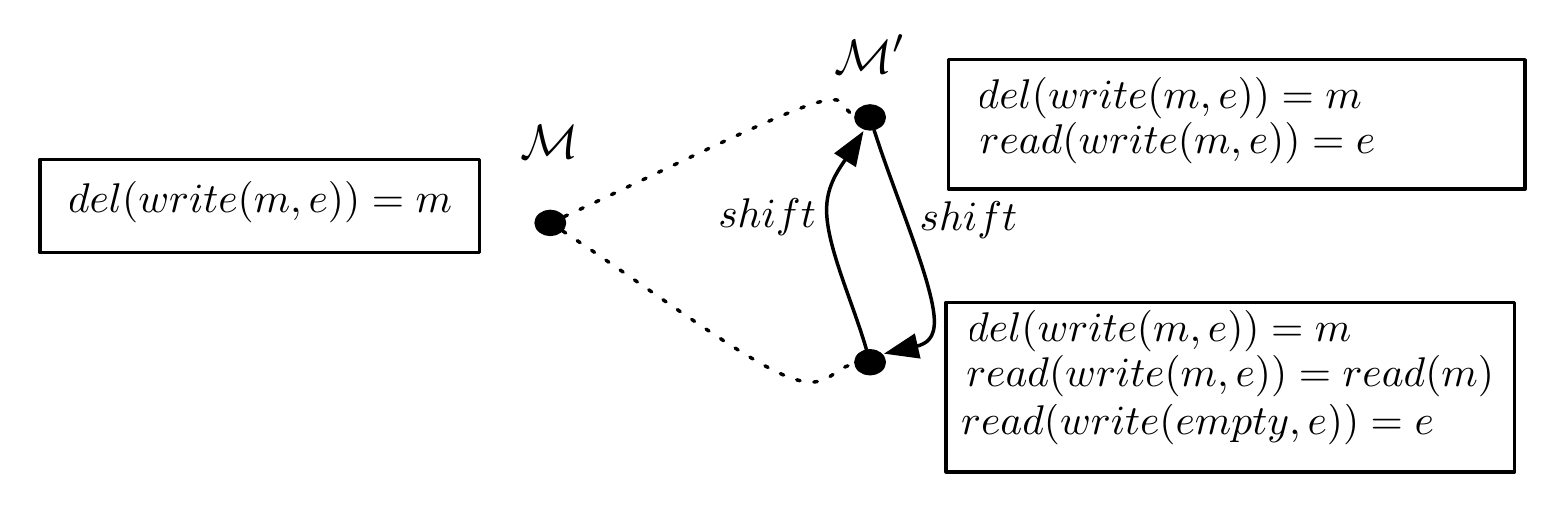}
\end{center}
\caption{Refinement  in $\HEQ$}
\end{figure}
It is not difficult to see that $R=\{(\star,s_1),(\star,s_2)\}$ is a
$\varphi$-refinement relation:  conditions (f.i) and (f.iii) are
trivially fulfilled and,  condition (f.ii) is a direct
consequence of properties representability of the initial models.
\end{example}

\section{Conclusions}

The paper introduced  notions of equivalence and refinement between models of hybridised logic specifications, i.e. specifications formalised in hybridised versions of  base logics used to describe  a systems' possible configurations.  The definition is parametric on precisely the base logic relevant for each application. 
Current work on this topic includes research on a full equivalence theorem, showing, in particular, in which cases $\HI$ logical equivalence entails bisimilarity. Another topic concerns   the study of typical constructions on Kripke structures (e.g. bounded morphism images, substructures and disjoint unions) and their characterisation under bisimilarity and refinement.

\subsection*{Acknowledgements}
Work funded by the ERDF through the Programme COMPETE and the Portuguese Government through FCT - Foundation for Science and Technology, under
contract \texttt{FCOMP-01-0124-FEDER-028923},
 \emph{Centro de Investiga\c{c}\~{a}o e Desenvolvimento em Matem\'atica e Aplica\c{c}\~{o}es} of Universidade de Aveiro,  and
   doctoral grant \texttt{SFRH/BDE/33650/2009} supported by FCT and
 \emph{Critical Software S.A., Portugal}.


\end{document}